%% file: main.tex
\documentclass[11pt,letterpaper]{article}
\usepackage{xspace}
\usepackage[utf8]{inputenc}
\usepackage[english]{babel}

\usepackage[dvipsnames]{xcolor}
\usepackage[colorlinks=true,pdfpagemode=UseNone,urlcolor=RoyalBlue,linkcolor=RoyalBlue,citecolor=OliveGreen,pdfstartview=FitH]{hyperref}
\definecolor{hkured}{HTML}{EE4123}
\definecolor{hkublue}{HTML}{009BD4}
\definecolor{hkugreen}{HTML}{00B38C}
\definecolor{hkuyellow}{HTML}{FED401}

\usepackage{framed}
\definecolor{shadecolor}{HTML}{E0E0E0}

\usepackage{tikz}
\usetikzlibrary{shapes, arrows.meta, positioning, patterns}
\usepackage{pgfplots}
\usepgfplotslibrary{fillbetween}
\pgfdeclarelayer{ft}
\pgfdeclarelayer{bg}
\pgfsetlayers{bg,main,ft}
\pgfplotsset{compat=1.15}

\usepackage{subcaption}
\usepackage{tabto}

\usepackage[margin=1in]{geometry}
\usepackage{amsfonts,amsmath,amsthm,amssymb}
\usepackage{mathtools,thmtools}
\usepackage{bm}

\usepackage{array}
\newcolumntype{x}[1]{>{\centering\arraybackslash\hspace{0pt}}p{#1}}

\usepackage{nicefrac}

\declaretheorem[parent=section]{theorem}
\declaretheorem[sibling=theorem]{lemma}

\declaretheorem[sibling=theorem]{remark}



\usepackage{todonotes}
\usepackage{tcolorbox}
\usepackage{booktabs}
\usepackage{caption}
    
\usepackage{multirow}

\usepackage{enumitem}

\usepackage[numbers, sort]{natbib}

\newtcolorbox{algorithm}[1]
{
	adjusted title = {#1},
	fonttitle = \bfseries,
  	beforeafter skip = 12pt,
}
\usepackage{cleveref}

\input{notation}

\title{Mixture-of-Experts Serving}
\author{
Zhiyi Huang \footnote{The University of Hong Kong. Email: zhiyi@cs.hku.hk, qinpeilou@connect.hku.hk}
\and
Qinpei Lou \footnotemark[1]
\and
Tao Xiao \footnote{Huawei Taylor Lab. Email: xiaotao21@huawei.com}
}
\date{July 2026}

\begin{document}

\begin{titlepage}
    \thispagestyle{empty}

    \maketitle

    \begin{abstract}
        \thispagestyle{empty}
        Mixture-of-Experts (MoE) models route each token to only a few expert networks, distributing the serving load across experts whose popularity shifts over time.
        A serving system must therefore dynamically decide how many GPUs to assign to each expert, trading off service latency against the cost of reconfiguring the assignment.
        We introduce a formal model of \emph{MoE Serving} and initiate a principled study of online and offline algorithms for it.
        Our main result is a polynomial-time $O(\sqrt{\log k})$-competitive online algorithm, where $k$ is the number of GPUs beyond one per expert.
        We complement it with a matching $\Omega(\sqrt{\log k})$ barrier for the online dual problem underlying our analysis.
        In the offline setting, we give a constant-factor approximation, show that MoE Serving is NP-hard, and rule out an FPTAS assuming ETH.
    \end{abstract}

\end{titlepage}

\section{Introduction}

Dense language models scale by making every token use more parameters, tying computation to capacity.
Mixture-of-Experts (MoE) models~\cite{JacobsEtAl-NC-1991,JordanJacobs-NC-1994} break this tie by maintaining a collection of expert networks but routing each token to only one or a few of them, so capacity can grow while the per-token computation stays low~\cite{ShazeerEtAl-ICLR-2017,FedusZS-JMLR-2022,DaiEtAl-arXiv-2024}.
The MoE architecture underlies prominent language models, including Gemini~3~Pro~\cite{GeminiThreePro-2025}, Mixtral~\cite{JiangEtAl-arXiv-2024}, DeepSeek-V3~\cite{DeepSeekAI-arXiv-2024}, and GPT-OSS~\cite{OpenAI-arXiv-2025}.
DeepSeek-V3, for instance, activates only $37$ billion of its $671$ billion parameters per token~\cite{DeepSeekAI-arXiv-2024}.

Sparsity keeps these models cheap to run, but it reshapes the inference problem.
The serving load is no longer handled by one dense model, but routed through experts whose popularity shifts over time and across application domains.
A serving system must therefore decide not only how many GPUs to provision, but also where to place that capacity among the experts as the routed workload changes.
To serve DeepSeek-V3, for example, the decoding stage deploys its $256$ routed experts across $320$ GPUs --- one expert per GPU, with the remaining GPUs holding redundant copies of the busiest experts.%
\footnote{We omit deployment details that are immaterial to our model, such as \emph{shared experts} and other placement optimizations; the qualitative picture is accurate.}
The reported DeepSeek-V3 deployment recomputes the placement periodically (e.g., every ten minutes) from recently observed load using simple heuristics~\cite{DeepSeekAI-arXiv-2024}.
Designing principled algorithms for efficient MoE serving is therefore an important research direction.

\subsection{MoE Serving Problem}

This paper introduces a theoretical model for MoE serving, and initiates the study of online and offline algorithms through competitive and approximation analysis.
Consider $m$ experts deployed on $n>m$ GPUs.
At each step $t$, the system observes the workload routed to each expert and decides how many GPUs to assign to it.

There are two main costs to consider.
The \emph{latency cost} at each step is the time it takes for all experts to finish their workloads, determined by the ratio of workload to the number of GPUs for the most heavily loaded expert.
The \emph{reconfiguration cost} is proportional to the number of GPUs that change expert assignment from the previous step, e.g., the I/O cost of copying neural network weights, warming caches, and other overheads.
We consider minimizing the sum of the two costs.%
\footnote{A weighted sum would appear more general, but any weight only rescales the workloads, which already range over all positive reals.
    We therefore omit weights for notational simplicity.}

The model deliberately abstracts away low-level details while keeping the main algorithmic tension.
Putting more GPUs on currently busy experts reduces latency, but frequent reallocation creates instability and reconfiguration overhead.

% We study both dynamic and static versions of the problem.  In the dynamic problem, the allocation may change over time and the algorithm must be online: at time $t$, it knows the current workload $r_t$ but not future workloads.  This captures elastic serving, where the system may adapt to changing traffic.  In the static problem, the allocation must remain fixed throughout the horizon.  This captures systems in which reconfiguration is prohibited, too expensive, or performed only at coarse time scales.  The static problem also gives a clean offline benchmark for understanding the combinatorial difficulty of expert placement.

\subsection{Algorithms and Barriers}

We parameterize the running time and competitive/approximation ratios of the algorithms by the number of experts $m$ and the number of GPUs $n$.
It is often more natural to consider the number of extra GPUs $k=n-m$, beyond the initial assignment of one GPU per expert.

    % Our main result is a polynomial-time online algorithm for MoE Serving with a sub-logarithmic competitive ratio.
    % We also show a natural barrier to improving the analysis of this algorithm.

    {
        \setlength{\fboxsep}{8pt}%
        \begin{snugshade}
            \noindent{\bfseries Online MoE Serving Results}\par
            \vspace{6pt}
            \begin{itemize}[topsep=0pt, itemsep=2pt]
                \item There is a polynomial-time $O(\sqrt{\log k})$-competitive online algorithm for MoE Serving.
                \item There is no $o(\sqrt{\log k})$-competitive algorithm for the online dual problem arising in our\\ primal-dual analysis.
            \end{itemize}
        \end{snugshade}
    }

We first design a fractional online algorithm, building on the now-standard techniques of online primal-dual with Fenchel duality~\cite{AzarEtAl-FOCS-2016,DevanurJain-STOC-2012,DevanurHuang-SODA-2014,HuangKim-SODA-2015} and Kullback-Leibler divergence regularization~\cite{BuchbinderCN-SODA-2014,BuchbinderGMN-SODA-2019,BubeckEtAl-STOC-2018,CoesterLee-COLT-2019}.
Applied directly, these techniques give only a $\log k$ competitive ratio.
Our analysis further exploits the smoothness of the latency cost function (doubling the assigned GPUs at best halves the cost) and sublinearity of its conjugate, to improve the competitive ratio to $O(\sqrt{\log k})$.

% The algorithm is a regularized greedy method.  Writing $f_t(x)=\max_i r_{t,i}/(1+x_i)$ for the processing cost at step $t$, it chooses at each step the fractional allocation minimizing $f_t$ plus a KL-divergence penalty from the previous allocation:
% \[
%     x_t
%     =
%     \argmin_{x\in\mathbb R_{\ge 0}^m:\,\sum_i x_i=k}
%     \left(
%     f_t(x)
%     +
%     \KL(x+\mathbf 1\,\|\,x_{t-1}+\mathbf 1)
%     \right).
% \]
% The shift $x+\mathbf 1$ is not merely notational: it is the number of GPUs actually assigned to an expert, including its mandatory baseline GPU, and it matches the reciprocal form of $f_t$.  The KL regularizer penalizes \emph{multiplicative} changes in expert capacity, which is natural because moving one GPU away from a lightly replicated expert matters more than moving one away from a heavily replicated one.

% We analyze the update by online primal-dual, fitting a feasible solution to an offset Fenchel dual of the relaxation.  The one non-routine step concerns the reconfiguration dual $\beta$: the KKT conditions of the update suggest movement increments $\beta_{t+1,i}-\beta_{t,i}\approx-\log\frac{x_{t,i}+1}{x_{t-1,i}+1}$, whose cumulative value can span an interval of order $\log(k+1)$ (as $x_{t,i}+1$ ranges from $1$ to $k+1$), overshooting the feasible dual range.  Rescaling the regularizer to compensate would already cost an $O(\log k)$ factor; instead we rescale the dual variables themselves and absorb the loss through smoothness inequalities for $f_t$ and its offset Fenchel conjugate, improving the degradation from $O(\log k)$ to $O(\sqrt{\log k})$.

For rounding, we first examine existing techniques and find that the online rounding schemes for caching and paging~\cite{BlumBK-FOCS-1999,BansalBN-SICOMP-2012} adapt to MoE serving.
These schemes, however, are computationally inefficient even in their original settings.
We exploit the smoothness of the latency cost function once again to design a polynomial-time Lazy Threshold Rounding that preserves both the expected latency and reconfiguration costs up to a constant factor.

Is the $\sqrt{\log k}$ factor inherent to MoE serving, or can it be improved?
Although we cannot answer this question unconditionally, we show that $\sqrt{\log k}$ is a barrier for our online primal-dual analysis.
It is both a blessing and a curse that an online primal-dual algorithm solves the primal and dual problems simultaneously, as the best achievable ratio is limited by the harder of the two.

The dual is an online budget-allocation problem with an unknown stopping time.
As requests keep arriving for a single expert, one must repeatedly decide how much dual mass to commit, with diminishing marginal payoffs.
The offline benchmark knows the stopping time and can spread this mass evenly across the steps.
The online algorithm does not, and is therefore forced to overcommit early.
This tension drives a $\sqrt{\log k}$-factor loss.

    % We also show that the $\sqrt{\log k}$ loss is inherent for a natural online dual-fitting framework.  More precisely, we define an online dual construction problem corresponding to the Fenchel dual used in the analysis and prove that no online algorithm for this dual construction problem can achieve competitive ratio better than
    % \[
    %     \Omega(\sqrt{\log k}).
    % \]
    % This should be interpreted as a limitation of the dual-construction method, not as an unconditional lower bound for Dynamic MoE Serving itself.  It explains why the analysis cannot be improved to a constant factor merely by a sharper accounting of the same online dual.

    % \subsection{Offline MoE Serving Results}

    {
        \setlength{\fboxsep}{8pt}%
        \begin{snugshade}
            \noindent{\bfseries Offline MoE Serving Results}\par
            \vspace{6pt}
            \begin{itemize}[topsep=0pt, itemsep=2pt]
                \item There is a polynomial-time constant-factor approximation algorithm for MoE Serving.
                \item There is no fully polynomial-time approximation scheme (FPTAS) for MoE Serving\\ assuming the Exponential Time Hypothesis (ETH).
            \end{itemize}
        \end{snugshade}
    }

% In the offline setting, the entire workload sequence $r_1, \dots, r_T$ is revealed in advance.  Here a constant-factor approximation for the full dynamic problem follows directly from our online machinery, and it is complemented by a matching hardness of approximation.

% \begin{informal}[Offline MoE Serving]
%     MoE Serving admits a polynomial-time constant-factor approximation, obtained by rounding the convex relaxation.  Its static special case admits a sharper $2$-approximation, and the underlying relaxation has integrality gap at least $2-1/m$.  Moreover, MoE Serving is NP-hard, and assuming ETH it admits no FPTAS.
% \end{informal}

The constant-factor approximation follows by solving the convex relaxation offline and rounding it, e.g., with the same scheme used by our online algorithm.
For the static special case, where the configuration must remain fixed throughout, e.g., when reconfiguration is prohibitively expensive, we give a simple $2$-approximation and an almost matching integrality gap of $2-\frac{1}{m}$.

% The constant-factor approximation is immediate: solve the convex relaxation \Primal offline --- which, being convex, can be solved to optimality in polynomial time --- and then apply the same threshold rounding used by our online algorithm.  Since the rounding preserves both the processing and reconfiguration costs up to constant factors, the result is an $O(1)$-approximation for offline (dynamic) MoE Serving.

% For the \emph{static} special case --- equivalently, MoE Serving with prohibitively large reconfiguration cost, so that a single configuration must serve the entire horizon --- we obtain a sharper and particularly simple $2$-approximation: solve the fractional relaxation, round every coordinate down, and allocate the remaining GPUs arbitrarily.  The analysis uses only the inequality
% \[
%     1+x \le 2(1+\lfloor x\rfloor).
% \]
% The factor $2$ is tight for this relaxation: with $m$ experts, $n=2m-1$ GPUs, and a single uniform workload vector, the fractional optimum can spread the $m-1$ extra GPUs evenly, while every integral solution leaves some expert without an extra GPU, producing an integrality gap of $2-1/m$.

Finally, we prove that MoE Serving is computationally hard even beyond this relaxation barrier.
By a reduction from the Densest $k$-Subgraph problem, we show that the static problem is already NP-hard.
Together with the ETH-hardness of distinguishing dense and sparse $k$-subgraphs~\cite{Manurangsi-STOC-2017}, the reduction also rules out an FPTAS.

\subsection{Further Related Work}

\paragraph{Online Algorithm Problems with Reconfiguration Costs.}
The reconfiguration cost alone has been studied, under various names, in the literature on online paging and caching~\cite{SleatorTarjan-CACM-1985,FiatEtAl-JAlg-1991,BansalBN-SICOMP-2012}, $k$-server~\cite{ManasseMS-JAlg-1990,KoutsoupiasP-JACM-1995}, and chasing convex bodies~\cite{FriedmanL-DCG-1993,BubeckLLS-STOC-2019,Sellke-SODA-2020}.
The sum of a reconfiguration cost and a state-dependent cost of specific forms has been studied in metrical task systems~\cite{BorodinLS-JACM-1992,BartalEtAl-STOC-1997} and online set cover with service costs~\cite{BuchbinderCN-SODA-2014}.
The state-dependent cost in these two problems, in their usual fractional formulations, is linear in the state,%
\footnote{Metrical task systems are often defined with an arbitrary cost over a finite set of states, which is equivalently a linear cost over distributions on those states, i.e., over $x \ge 0$ with $\|x\|_1 = 1$.}
whereas our latency cost is nonlinear.
The best competitive ratios for all the above problems are logarithmic or worse, except in special cases.

Smoothed online convex optimization~\cite{AndrewEtAl-COLT-2013,BansalEtAl-APPROX-2015,ChenGW-COLT-2018} and convex function chasing~\cite{ArgueGG-COLT-2020} study the sum of a reconfiguration cost and a general convex state-dependent cost, capturing the fractional relaxation of MoE serving as a special case.
For arbitrary convex costs, the competitive ratio depends polynomially on the dimension~\cite{ChenGW-COLT-2018,ArgueGG-COLT-2020}.
Improved algorithms and competitive ratios are only known under structural assumptions that MoE serving does not satisfy, including one-dimensional decisions~\cite{BansalEtAl-APPROX-2015}, locally polyhedral costs~\cite{ChenGW-COLT-2018}, or strongly convex and smooth costs~\cite{ArgueGG-COLT-2020}.

% Smoothed online convex optimization (SOCO) and its generalization, convex function chasing, are closest in spirit: at each step the algorithm observes a convex hitting cost $f_t$ over $\R^d$, picks a point $x_t$, and pays $f_t(x_t) + \|x_t - x_{t-1}\|$, with chasing convex bodies the special case where $f_t$ is $0$ on a convex body and $+\infty$ outside.
% Constant-competitive algorithms are known only in structured regimes --- $2$-competitive on the line~\cite{BansalEtAl-APPROX-2015}, and dimension-free constants in high dimensions when the hitting costs are locally polyhedral, via online balanced descent~\cite{ChenGW-COLT-2018}, or $\kappa$-well-conditioned, with ratio $O(\sqrt{\kappa})$~\cite{ArgueGG-COLT-2020} --- while general high-dimensional costs require a $\mathrm{poly}(d)$ ratio.
% Our problem is a discrete instance of this template with special structure: the hitting cost is the inverse-load bottleneck $\max_{i \in [m]} r_{t,i}/(1+x_{t,i})$, the movement cost is the $\ell_1$ distance, and the decision is an \emph{integral} allocation of a fixed budget of $k$ GPUs over a simplex rather than an arbitrary point of $\R^d$.
% This structure calls for a different algorithm: instead of the Euclidean projection and balanced-descent updates that these methods tailor to smooth or polyhedral costs, we run an entropy-regularized greedy update suited to the simplex and the reciprocal cost, followed by an online rounding step to recover an integral allocation.

\paragraph{MoE Expert Imbalance and Placement.}
\citet{GoMahajan-arXiv-2025} optimized a static expert placement by solving an integer program.
\citet{LiEtAl-ATC-2023} instead showed empirically that expert load is skewed and shifts over time, so dynamically reallocating to the predicted load helps.
However, they did not consider the reallocation cost, such as the overhead for copying the experts' weights mentioned in the DeepSeek-V3 report~\cite{DeepSeekAI-arXiv-2024}.
This paper initiates the principled study of the trade-off between service latency and reconfiguration overhead.

\paragraph{Other Challenges in Language Model Serving.}
Serving language models raises a range of scheduling problems.
The natural scheduling granularity is a token-generation iteration, and high GPU utilization relies on effectively batching the active requests~\cite{YuEtAl-OSDI-2022,AgrawalEtAl-OSDI-2024}.
A further difficulty is that a request's output length is unknown when it arrives~\cite{WuEtAl-NSDI-2026}.
Compounding this, the key-value (KV) cache usage grows with the output length and can become the dominant memory constraint~\cite{KwonEtAl-SOSP-2023}.

Classical scheduling theory considered batching and unknown job lengths separately, well before the language model era.
The former has been studied for objectives such as makespan~\cite{Uzsoy-IJPR-1994}, completion time~\cite{Uzsoy-IJPR-1994,ChandruLU-IJPR-1993}, and flow time~\cite{HochbaumL-OR-1997}.
The latter was commonly known as non-clairvoyant scheduling~\cite{MotwaniPT-TCS-1994,KalyanasundaramP-JACM-2000,BecchettiL-JACM-2004}.

A recent line of theoretical work studied scheduling under a growing KV cache usage, in combination with batching and unknown job lengths.
\citet{WangYZ-arXiv-2025} proved NP-hardness of the offline problem with heterogeneous prefill and decode lengths, and gave a constant-factor approximation.
\citet{JailletEtAl-arXiv-2025} studied the online problem against a hindsight-optimal integer program and obtained a constant competitive ratio under further conditions.
\citet{KongQYZ-arXiv-2026} exploited the spatio-temporal geometry of cache growth to sharpen the competitive ratio toward $3$ in the high-concurrency regime.
% This line treats KV-cache memory as the scarce resource and optimizes the request schedule, whereas our model is complementary, treating the allocation of GPU capacity across experts --- and the reconfiguration cost of changing it over time --- as the scarce resource.

% Unlike these classical batch-processing models, LLM batching is coupled to autoregressive token generation, phase asymmetry, and dynamically changing KV-cache memory.

% % \paragraph{Unknown Generation Lengths.}
% %
% In language model serving, a request's output length, and hence its future resource demand, is unknown when the request arrives.
% FastServe preempted inference at token granularity and used a multilevel-feedback-queue scheduler to reduce head-of-line blocking~\cite{WuEtAl-NSDI-2026}, while others learned to predict the relative order of output lengths to approximate shortest-job-first scheduling~\cite{FuEtAl-NeurIPS-2024}.

\section{Preliminaries}

\paragraph{Notation.}
Let $\Z$ and $\R$ denote the sets of integers and real numbers, respectively.
We use subscripts $\ge 0$ and $> 0$ to indicate the sets of non-negative and positive numbers.
For example, $\Z_{\ge 0}$ and $\Z_{> 0}$ are the sets of non-negative and positive integers, respectively.
For any positive integer $m$, let $[m]$ denote the set $\{1, 2, \dots, m\}$.

\paragraph{Dynamic Model.}
We consider the \emph{Dynamic MoE Serving} problem, where we deploy an MoE LLM consisting of $m$ experts onto $n > m$ GPUs.
The system processes requests over $T$ discrete time steps.
At each step $t \in [T]$, the algorithm receives a request represented as a workload vector $r_t = (r_{t,i})_{i \in [m]} \in \R_{\ge 0}^m$, where $r_{t,i}$ represents the volume of user requests routed to expert $i$.

Upon observing $r_t$, the algorithm picks a deployment configuration $c_t = (c_{t,i})_{i \in [m]} \in \Z_{> 0}^m$, where $c_{t,i}$ is the number of GPUs allocated to expert $i$.
The constraint $c_{t,i} \ge 1$ ensures that every expert is available at all times.
This effectively means that the first $m$ GPUs are allocated evenly to the experts.
Analytically, it is more convenient to track the allocation of the remaining GPUs.
Let $x_{t,i} = c_{t,i} - 1 \in \Z_{\ge 0}$ denote the number of remaining GPUs assigned to expert $i$, subject to the budget constraint that $\sum_{i \in [m]} x_{t,i} = k \defeq n - m$.

Given a configuration $x_t$, the \emph{latency cost} of the system at step $t$ is determined by the time it takes for all experts to complete their workloads:
\[
    \max_{i \in [m]} \frac{r_{t,i}}{1 + x_{t,i}}
    ~.
\]

Furthermore, transitioning from the previous configuration $x_{t-1}$ to the new configuration $x_t$ incurs a \emph{reconfiguration cost} proportional to the $L_1$ distance between the allocations:
\[
    \sum_{i \in [m]} \big| x_{t-1,i} - x_{t,i} \big|
    ~,
\]
where $x_0$ is an arbitrary initial configuration given as part of the instance.

The objective of the Dynamic MoE Serving problem is to minimize the sum of the total latency cost and the total reconfiguration cost over the time horizon $[T]$:
\[
    \sum_{t \in [T]} \bigg( \max_{i \in [m]} \frac{r_{t,i}}{1 + x_{t,i}} + \sum_{i \in [m]} \big| x_{t-1,i} - x_{t,i} \big| \bigg)
    ~.
\]

\paragraph{Static Model.}
The \emph{Static MoE Serving} problem is the variant where reconfiguration is strictly prohibited, meaning $x_t = x$ for all $t \in [T]$.
Because the reconfiguration cost is zero, the problem simplifies to finding a single, fixed configuration $x \in \Z_{\ge 0}^m$ subject to  $\sum_i x_i = k$ that minimizes the total latency cost:
\[
    \sum_{t = 1}^T \max_{i \in [m]} \frac{r_{t,i}}{1 + x_i}
    ~.
\]

\paragraph{Approximation and Competitive Ratios.}
Let $\OPT$ denote the minimum total cost achieved by the optimal schedule in hindsight.
Let $\ALG$ denote the total cost incurred by the algorithm's schedule.
For any $\Gamma \ge 1$, an offline algorithm yields a $\Gamma$-approximation if $\E \big[ \ALG \big]\le \Gamma \cdot \OPT$ for all instances. We analyze offline approximation algorithms for both the static and dynamic models.
For any $\Gamma \ge 1$, an online algorithm is $\Gamma$-competitive if $\E \big[ \ALG \big] \le \Gamma \cdot \OPT + C$ for all instances, where $C$ is an additive term independent of the time horizon $T$ (but can depend on the numbers of experts and GPUs). We consider online algorithms only in the dynamic model.

\section{Convex Programs}
\label{sec:convex-programs}

% \subsection{Primal Convex Program Relaxation}

This section derives the convex program relaxation of the Dynamic MoE Serving problem and its Fenchel dual.
The proofs of the lemmas are standard, and therefore deferred to \Cref{app:convex-programs}.

For each step $t \in [T]$, define a convex function $f_t$ to denote the latency cost of configuration $x_t$, extended to non-negative real vectors:
\[
    f_t(x_t) ~\defeq~
    \begin{cases}
    \max_{i \in [m]} \dfrac{r_{t,i}}{1+x_{t,i}} & \mbox{if $x_{t,i} \ge 0$ for all $i \in [m]$;} \\
    +\infty & \mbox{otherwise.}
    \end{cases}
\]

To linearize the reconfiguration cost, we let $y_{t,i} \ge \max \{ x_{t,i} - x_{t-1,i}, 0 \}$ denote the incremental cost of expert $i$ from step $t-1$ to step $t$.
The budget $k$ is constant at all times (i.e., $\sum_i x_{t,i} = \sum_i x_{t-1,i} = k$).
Hence, the incremental cost matches the decremental cost, and equals half the total reconfiguration cost.

Putting together, we have the following convex program relaxation of Dynamic MoE Serving, and denote its objective value as $\Primal$ and the optimal value as $\Primal^\star$.
\begin{align*}
    \textrm{minimize} \quad
                                                                          &
    \sum_{t \in [T]} f_t(x_t) + \sum_{t \in [T]} \sum_{i \in [m]} y_{t,i} &   & \text{(\Primal)}                     \\
    \textrm{subject to} \quad
                                                                          &
    \sum_{i \in [m]} x_{t,i} \le k                                          &   & \forall t \in [T]                    \\
                                                                          &
    y_{t,i} \ge x_{t,i} - x_{t-1,i}                                       &   & \forall i \in [m], \forall t \in [T] \\[2ex]
                                                                          &
    x_{t,i}, y_{t,i} \ge 0                                                &   & \forall i \in [m], \forall t \in [T]
\end{align*}

We do not attempt to optimize the constants in this paper, and thus, omit the factor $2$ difference between incremental and total reconfiguration costs for ease of notation. 
There is an additional technical relaxation. Although all algorithms we design maintain the equality-budget constraint $\sum_i x_{t,i}=k$ at every time step, in the primal relaxation used for the dual analysis we relax this equality by $\sum_i x_{t,i}\le k$. This can only decrease the fractional optimum, so the relaxation remains a valid lower bound on the optimal integral cost. The inequality form is more convenient for our argument.
% Besides, we use the inequality $\sum_i x_{t,i} \le k$ rather than the equality constraint $\sum_i x_{t,i}=k$, as a relaxation. It may only decrease the fractional optimum, and hence still gives a valid lower bound on the integral offline optimum. This relaxation is convenient for our later analysis, since the dual variable associated with the resource constraint can be taken to be non-negative.

\begin{lemma}
    \label{lem:primal-objective}
    For any instance of Dynamic MoE Serving, we have $\Primal^\star \le \OPT$.
\end{lemma}

% \subsection{Fenchel Dual Convex Program}
% \label{sec:dual}
% \label{sec:dual}

We use an offset Fenchel dual of the primal relaxation.
For each $t \in [T]$, define the function $\hat{f}_t : \R_{\ge 0}^m \to \R \cup \{-\infty\}$ with respect to a non-negative dual variable $\gamma_t \ge 0$:
\[
    \hat{f}_t(\gamma_t) ~\defeq~ \inf_{x_t \ge 0} \big( f_t(x_t) + \langle x_t, \gamma_t \rangle \big)
    ~.
\]

The standard convex conjugate of $f_t$ is $f_t^\star (\gamma_t) = \sup_{x_t \ge 0} \big(\langle x_t, \gamma_t \rangle - f_t(x_t)\big)$.
The above convention is equivalent to $\hat{f}_t(\gamma_t) = - f_t^\star(-\gamma)$, which keeps $\gamma_t$ and $\hat{f}_t(\gamma_t)$  non-negative.
With this definition, we recover $f_t$ as the pointwise supremum of its affine minorants:
\[
    f_t(x_t) \,=\, \sup_{\gamma_t \ge 0} \big( \hat{f}_t(\gamma_t) - \langle x_t, \gamma_t \rangle \big)
    ~.
\]

\begin{remark}[Subgradient Convention]
    We use $\partial f_t(x_t)$ to denote the subdifferential of the finite latency-cost expression
    \[
        x_t \mapsto \max_{i\in[m]} \frac{r_{t,i}}{1+x_{t,i}}
    \]
    on its natural domain $x_{t,i} > -1$.
    The normal cone from the constraint $x_t\ge 0$ is handled separately by the convex program.
    With this convention, $\partial f_t(x_t)$ is the convex hull of the gradients of the active coordinates attaining the maximum.
\end{remark}

\begin{lemma}
    \label{lem:property-of-f}
    For any $t \in [T]$, the following properties hold.
    \begin{enumerate}
        \item (Conjugate Pair) For any
              $\gamma_t\in -\partial f_t(x_t)$, $f_t(x_t) = \left\langle \gamma_t, x_t+\mathbf 1\right\rangle$.
        \item (Smoothness of $f_t$) For any $\theta \ge 1$ and $x_t, x'_t \in \R_{\ge 0}^m$ satisfying $x_{t,i}+1 \le \theta (x'_{t,i}+1)$ for every $i \in [m]$, we have $f_t(x'_t) \le \theta f_t(x_t)$.
        \item (Sublinearity of $\hat{f}_t$) For any $\theta \in (0,1]$ and $\gamma_t \ge 0$, $\hat f_t(\theta \gamma_t) \ge \sqrt{\theta}\,\hat f_t(\gamma_t)$.
    \end{enumerate}
\end{lemma}

The dual variables are $\alpha_t \in \R_{\ge 0}$ for the budget constraints, $\beta_{t,i} \in \R_{\ge 0}$ for the incremental reconfiguration constraints, and $\gamma_{t,i} \in \R_{\ge 0}$ for the Fenchel representation of $f_t$.
We also use an auxiliary terminal variable $\beta_{T+1} \in \R_{\ge 0}^m$ to write the dynamic constraints uniformly.
The resulting offset Fenchel dual program is:
\begin{align*}
    \textrm{maximize} \quad
                                                   &
    \sum_t \hat{f}_t(\gamma_t) - \sum_t k \alpha_t &                                                         & \text{(\Dual)}                                          \\
    \textrm{subject to} \quad
                                                   & \gamma_{t,i} + \beta_{t+1,i} - \beta_{t,i} \le \alpha_t &                & \forall t \in [T], \forall i \in [m]   \\[1.5ex]
                                                   & 0 \le \beta_{t,i} \le 1                                 &                & \forall t \in [T+1], \forall i \in [m] \\[1.5ex]
                                                   & \gamma_{t,i} \ge 0                                      &                & \forall t \in [T], \forall i \in [m]
                                                   \\[1.5ex]
                                                   & \alpha_t \ge 0                                      &                & \forall t \in [T]
\end{align*}

Let $\Dual$ denote the objective value of the above program and the optimal value as $\Dual^\star$.
We relate the optimal primal and dual objective values as follows.

\begin{lemma}
    \label{lem:primal-dual-objective}
    For any instance of Dynamic MoE Serving, we have $\Dual^\star \le \Primal^\star + k$.
\end{lemma}

\paragraph{Karush-Kuhn-Tucker Conditions.}
While we will not use the optimality conditions of the optimal primal and dual solutions directly, we include them below for easy comparison with the optimality conditions of the projection steps of our online algorithm.
Let $x^\star, y^\star$ and $\alpha^\star, \beta^\star, \gamma^\star$ be the optimal primal and dual solutions.
They satisfy not only the primal and dual constraints, but also the following optimality conditions, including stationarity and complementary slackness:
\begin{align*}
    -\gamma^\star_t                                                                                         & ~\in~ \partial f_t(x^\star_t) ~, \\%[1ex]
    \beta^\star_{t,i} \big( y^\star_{t,i} - x^\star_{t,i} + x^\star_{t-1,i} \big)                           & ~=~ 0 ~,                         \\%[1ex]
    y^\star_{t,i} \big( 1 - \beta^\star_{t,i} \big)                                                         & ~=~ 0 ~,                         \\%[1ex]
    x^\star_{t,i} \big( \alpha^\star_t + \beta^\star_{t,i} - \beta^\star_{t+1,i} - \gamma^\star_{t,i} \big) & ~=~ 0 ~, \\%[1ex]
    \alpha^\star_t \big( \sum_{i \in [m]} x^\star_{t,i} - k \big) & ~=~ 0 ~.                         
\end{align*}

\section{Dynamic Model}
\label{sec:dynamic}

This section considers the Dynamic MoE Serving problem, for which we give efficient online and offline algorithms.
% The main result is the following theorem.

\begin{theorem}
    \label{thm:dynamic}
    The Dynamic MoE Serving problem:
    \begin{enumerate}
        \item There is a polynomial-time $O(\sqrt{\log k})$-competitive online algorithm; and
        \item There is a polynomial-time $O(1)$-approximation algorithm.
    \end{enumerate}
\end{theorem}

We devote the rest of this section to proving \Cref{thm:dynamic}.
\Cref{sec:fractional-online-alg} gives a regularization-based fractional algorithm and derives the stated competitive ratio of $O(\sqrt{\log k})$ (\Cref{lem:fractional-competitive}).
% This algorithm performs an online mirror-descent update with entropy regularization to balance the processing time against the reconfiguration cost.
% \Cref{sec:competitive-analysis1} analyzes the competitive ratio of the fractional algorithm.

\Cref{sec:online-rounding} then presents an online rounding scheme that preserves both the latency cost and reconfiguration cost of the fractional algorithm up to a constant factor (\Cref{lem:rounding-competitive}).
Combining the two lemmas proves the online part of \Cref{thm:dynamic}; for the offline approximation, we solve the equality-budget version of the convex relaxation, i.e., with $\sum_i x_{t,i}=k$ for all $t$, and then apply \Cref{lem:rounding-competitive}.

\subsection{Regularized Greedy Algorithm}
\label{sec:fractional-online-alg}

Let $\eta = \max\{1, \ln(k+1)\}$ be the regularization parameter, chosen to optimize the competitive ratio from our analysis.
The regularizer is the generalized Kullback-Leibler (KL) divergence:
\[
    \KL(u \| v) ~\defeq~ \sum_{i \in [m]} \left( u_i \ln \frac{u_i}{v_i} - u_i + v_i \right) ~.
\]

For the configuration vectors in dynamic MoE serving, we use $u_i = x_{t,i} + 1$ and $v_i = x_{t-1,i} + 1$. In our algorithm, we use $\sum_i x_{t,i}=k$ as our budget constraint rather than $\sum_i x_{t,i} \le k$.
For every step $t \in [T]$, the linear terms cancel out since the total auxiliary budget is constant, i.e., $\sum_{i \in [m]} x_{t,i} = \sum_{i \in [m]} x_{t-1,i} = k$.
Thus, the divergence simplifies to the standard relative entropy.
% For the initial step $t=1$, the linear terms evaluate to a constant $k$ (since $x_0$ is all-zero) and do not affect the optimization.

\begin{algorithm}{Regularized Greedy Algorithm}
    At every step $t \in [T]$, select a feasible fractional configuration $x_t$ that minimizes the sum of the current latency cost and the divergence from the previous configuration:
    \[
        % \begin{equation}
        % \label{eqn:reg-greedy1}
        x_t ~=~ \argmin_{x \in \R_{\ge 0}^m \,:\, \sum_{i \in [m]} x_i = k} \Bigl( f_t(x) +  \KL(x + \mathbf{1} \,\|\, x_{t-1} + \mathbf{1}) \Bigr)
        ~.
        % \end{equation}
    \]
\end{algorithm}

% \subsection{Competitive Analysis}
% \label{sec:competitive-analysis1}
\label{sec:competitive-analysis1}

% Next, we show that this algorithm achives the desired competitive ratio in the relaxed fractional setting, via an online primal dual analysis, i.e., by constructing a feasible solution to the dual program and bounding the ratio between the primal and dual objectives.

\begin{lemma}
    \label{lem:fractional-competitive}
    Regularized Greedy is $O(\sqrt{\log k})$-competitive.
\end{lemma}

\begin{proof}
    % The online primal-dual analysis proceeds in three steps: (i) deriving the optimality conditions for the algorithm's update, (ii) constructing a feasible dual solution based on the optimality conditions, and (iii) bounding the primal objective by the dual counterpart.
    %
    % \paragraph{KKT Conditions of Regularized Greedy Updates.}
    %
    \textbf{KKT Conditions of Regularized Greedy.}
    The update step is a convex minimization problem over the simplex-like constraint set $\big\{ x \in \R_{\ge 0}^m \mid \sum_{i \in [m]} x_i = k \big\}$. The Lagrangian is:
    \[
        \mathcal{L}_t(x_t, \alpha^{\prime}_t) ~=~ f_t(x_t) +  \sum_{i \in [m]} \biggl( (x_{t,i}+1) \ln \frac{x_{t,i}+1}{x_{t-1,i}+1} - x_{t,i} + x_{t-1,i} \biggr) + \alpha^{\prime}_t \biggl( \sum_{i \in [m]} x_{t,i} - k \biggr) ~,
    \]
    where $\alpha^{\prime}_t \in \R$ is the Lagrange multiplier for the budget constraint.
    The KKT conditions imply that at the optimum $x_t$, there are $\alpha^{\prime}_t \in \R$ and $\gamma^{\prime}_t \in \R_{\ge 0}^m$ such that:
    \begin{align}
        - \gamma^{\prime}_t                                                             & ~\in~ \partial f_t(x_t) ~,                        \\
        - \gamma^{\prime}_{t,i} + \ln \frac{x_{t,i}+1}{x_{t-1,i}+1} + \alpha^{\prime}_t & ~\ge~ 0                    & \forall i \in [m] ~,
        \label{eqn:kkt}
    \end{align}
    where the latter holds with equality if $x_{t,i} > 0$.

    \paragraph{Dual Construction and Feasibility.}
    %
    % Let $\alpha^{\prime}_t$ be the minimax solution from solving the Lagrangian of the Regularized Greedy decision.
    % Let $\gamma^{\prime}_t \defeq - \nabla f_t(x_t)$, which is a non-negative vector since $f_t$ is monotone decreasing.
    % We obtain that for any $i \in [m]$:
    % %
    % \begin{equation} \label{eq:kkt_alg1}
    %     \gamma^{\prime}_{t,i} -\ln \frac{x_{t,i}+1}{x_{t-1,i}+1} ~\le~ \alpha^{\prime}_t
    %     ~,
    % \end{equation}
    % %
    % where it holds with equality if $x_{t,i} > 0$.
    Comparing \Cref{eqn:kkt} with the first dual constraint, $\gamma_{t,i} + \beta_{t+1,i} - \beta_{t,i} \le \alpha_t$, it is natural to consider $\alpha_t = \alpha^{\prime}_t$, $\gamma_{t,i} = \gamma^{\prime}_{t,i}$, and $\beta_{t+1,i} - \beta_{t,i} = - \ln \frac{x_{t,i}+1}{x_{t-1,i}+1}$.
    However, this is infeasible in general, because the cumulative changes in $\beta$ may be as large as $\ln(k+1)$ when $x_{t,i}$ ranges from $0$ to $k$ over time, violating the bounds $0 \le \beta_{t,i} \le 1$.

    The standard fix in the regularized greedy framework (e.g., \cite{BuchbinderCN-SODA-2014,BuchbinderGMN-SODA-2019}) is to scale the regularizer.
    Indeed, changing the update step to minimize $f_t(x) + \frac{1}{\ln(k+1)} \KL(x + \mathbf{1} \,\|\, x_{t-1} + \mathbf{1})$ would make the above dual construction feasible, but result in an inferior competitive ratio of $O(\log(k+1))$.

    Instead, we directly scale the dual variables to restore dual feasibility, and exploit the smoothness of $f_t$ and the sublinearity of $\hat{f}_t$ to bound the loss in the competitive ratio.
    Let:
    \[
        \alpha_t \,\defeq\, \frac{1}{\eta} \alpha^{\prime}_t
        ~,\qquad
        \beta_{t,i} \,\defeq\, - \frac{1}{\eta} \ln \frac{x_{t-1,i} + 1}{k+1}
        ~,\qquad
        \gamma_t \,\defeq\, \frac{1}{\eta} \gamma^{\prime}_t
        ~.
    \]

    Then, \Cref{eqn:kkt} becomes identical to the first dual constraint:
    \begin{equation}
        \label{eqn:complementary-slackness}
        \gamma_{t,i} + \beta_{t+1,i} - \beta_{t,i} \le \alpha_t
        ~,
    \end{equation}
    with equality if $x_{t,i} > 0$.
    Meanwhile, $0 \le \beta_{t,i} \le 1$ follows from the definition and $0 \le x_{t,i} \le k$.

    % Scale the the second condition of \Cref{eqn:kkt} by a factor $\eta$, we can get
    % \begin{equation} \label{eq:scaled_kkt_alg1}
    %     \frac{1}{\eta} \gamma^{\prime}_{t,i} -\frac{1}{\eta} \ln \frac{x_{t,i}+1}{x_{t-1,i}+1} ~\le~ \frac{1}{\eta}\alpha^{\prime}_t
    %     ~.
    % \end{equation}

    Finally, we show that:
    \begin{equation}
        \label{eqn:alpha-nonegative}
        \alpha_t ~\ge~ 0
        ~,
    \end{equation}
    %
    % even though it is not imposed as a constraint in the original problem or the update step.
    which satisfies the last constraint of original dual program.
    Since $\|x_t\|_1 = \|x_{t-1}\|_1 = k$ following the algorithm, we have $x_{t,i} \le x_{t-1,i}$ for some $i$.
    Combining with $\gamma^{\prime}_{t,i} \ge 0$ and \Cref{eqn:kkt}, we get that $\alpha^{\prime}_t$ and thus $\alpha_t$ are non-negative.

    % Further, compare \Cref{eq:scaled_kkt_alg1} with the first dual constraint:
    % %
    % \begin{equation}
    % 	\label{eqn:dual-constraint1}
    % 	\gamma_{t,i} + \beta_{t+1,i} - \beta_{t,i} ~\le~ \alpha_t
    % 	~.
    % \end{equation}

    % Naturally, we want to let $\beta_{t+1,i} - \beta_{t,i} = - \frac{1}{\eta}\ln \frac{x_{t,i}+1}{x_{t-1,i}+1}$.
    % Combining with the ranges of the variables, i.e., $0 \le x_{t,i} \le k$ and $0 \le \beta_{t,i} \le 1$, we derive:
    % %
    % \[
    %     \beta_{t,i} \,\defeq\, - \frac{1}{\eta} \ln \frac{x_{t-1,i} + 1}{k+1} ~.
    % \]

    % Also, we let
    % %
    % \[
    %     \alpha_t = \frac{1}{\eta} \alpha^{\prime}_t, \quad \gamma_{t,i} = \frac{1}{\eta} \gamma^{\prime}_{t,i}~.
    % \]

    % In sum, the dual assignment defined above is feasible.
    % The algorithm's KKT condition in \Cref{eq:scaled_kkt_alg1} matches the dual constraint \eqref{eqn:dual-constraint1}, and holds with equality if $x_{t,i} > 0$.

    \paragraph{Comparing the Primal and Dual Objectives.}
    Recall that:
    \[
        \Primal = \sum_{t \in [T]} \Bigl( f_t(x_t) + \sum_{i \in [m]} y_{t,i} \Bigr)
        ~,\qquad
        \Dual ~=~ \sum_{t \in [T]} \big( \hat{f}_t(\gamma_t) - k\alpha_t \big) ~,
    \]
    where $y_{t,i} = \max(0, x_{t,i} - x_{t-1,i})$.
    The regularized greedy update ensures that the reconfiguration cost is bounded by the latency cost.
    This standard proof is deferred to \Cref{app:balance-cost}.

    \begin{lemma}
        \label{lem:balance-cost}
        For every step $t \in [T]$, $\sum_i y_{t,i} \le f_t(x_t)$.
    \end{lemma}

    % \begin{align*}
    %     \sum_i y_{t,i}
    %     % &
    %     % ~\le~ \eta \sum_{i \,:\, x_{t,i} > x_{t-1,i}} (x_{t,i}+1)(\gamma_{t,i} - \alpha_t) \\
    %     % &
    %     % ~\le~ \eta \sum_{i \,:\, x_{t,i} > x_{t-1,i}} (x_{t,i}+1) \gamma_{t,i} \tag{$\alpha_t \ge 0$} \\
    %     % &
    %     % ~\le~ \eta \sum_{i \in [m]} (x_{t,i}+1) \gamma_{t,i} \\
    %     % % \tag{$\gamma_{t,i} \ge 0$} \\
    %     &
    %     ~\le~  \sum_{i \in [m]} (x_{t,i}+1) \gamma^{\prime}_{t,i} \\
    %     % \tag{$\gamma_{t,i} = \frac{1}{\eta} \gamma^{\prime}_{t,i}$} \\
    %     &
    %     ~=~ \langle - \nabla f_t(x_t), x_t + \mathbf{1} \rangle\\
    %     &
    %     ~=~  f_t(x_t)
    %     \tag{using \Cref{lem:property-of-f}}
    %     ~.
    % \end{align*}

    Now consider the dual objective.
    Since the algorithm guarantees $\sum_i x_{t,i} = k$ and the KKT \Cref{eqn:complementary-slackness}, we rewrite the second term of the dual objective as:
    \begin{align*}
        \sum_{t \in [T]} k \alpha_t
         &
        ~=~
        \sum_{t \in [T]} \sum_{i \in [m]} x_{t,i} \big( \gamma_{t,i} + \beta_{t+1,i} - \beta_{t,i} \big)
        ~.
    \end{align*}

    Further by $\gamma_t = \frac{1}{\eta} \gamma^{\prime}_t$, the dual objective equals:
    \[
        \sum_{t \in [T]} \biggl( \hat{f}_t \Bigl(\frac{1}{\eta} \gamma^{\prime}_t\Bigr) - \Bigl\langle x_t, \frac{1}{\eta} \gamma^{\prime}_t \Bigr\rangle \biggr) + \sum_{t \in [T]} \langle x_t, \beta_t - \beta_{t+1} \rangle
        ~.
    \]

    By the sublinearity of $\hat{f}_t$ (\Cref{lem:property-of-f}) and $\gamma'_t \ge 0$, the first part can be bounded by:
    \[
        \frac{1}{\sqrt{\eta}} \sum_{t \in [T]} \Big( \hat{f}_t(\gamma^{\prime}_t) - \langle x_t, \gamma^{\prime}_t \rangle \Big) = \frac{1}{\sqrt{\eta}} \sum_{t \in [T]} f_t(x_t)
        ~.
    \]
    % \begin{align*}
    % 	\sum_{t \in [T]} \Big( \hat{f}_t(\frac{1}{\eta} \gamma^{\prime}_t) - \langle x_t, \frac{1}{\eta} \gamma^{\prime}_t \rangle \Big) + \sum_{t \in [T]} \langle x_t, \beta_t - \beta_{t+1} \rangle
    % 	&
    % 	~\geq~ \sum_{t \in [T]} \Big(\frac{1}{\sqrt{\eta}} \hat{f}_t(\gamma^{\prime}_t) - \frac{1}{\sqrt{\eta}}\langle x_t, \gamma^{\prime}_t \rangle \Big)+ \sum_{t \in [T]} \langle x_t, \beta_t - \beta_{t+1} \rangle\\
    %     &
    %     ~=~
    % 	\frac{1}{\sqrt{\eta}}\sum_{t \in [T]} f_t(x_t) + \sum_{t \in [T]} \langle x_t, \beta_t - \beta_{t+1} \rangle
    % 	~.
    % \end{align*}

    The second term is bounded by telescoping:
    \begin{align*}
        \sum_{t \in [T]} \langle x_t, \beta_t - \beta_{t+1} \rangle
         &
        ~=~
        \sum_{t \in [T]} \langle x_t + \mathbf{1}, \beta_t - \beta_{t+1} \rangle + \langle \mathbf{1}, \beta_{T+1} - \beta_1 \rangle \\
         &
        ~=~
        \frac{1}{\eta} \sum_{t \in [T]} \KL(x_t+\mathbf{1} \,\|\, x_{t-1}+\mathbf{1})  + \langle \mathbf{1}, \beta_{T+1} - \beta_1 \rangle
        ~.
    \end{align*}

    The KL-divergence is non-negative; the rest is at least $-m$, as $\beta_{T+1,i} \ge 0$ and $\beta_{1,i} \le 1$ for all $i$.

    Putting everything together, we conclude that the dual objective satisfies:
    \[
        \Dual ~\ge~ \frac{1}{\sqrt{\eta}}\sum_{t \in [T]} f_t(x_t) - m \ge \frac{1}{2\sqrt{\eta}} \Primal - m
        ~.
    \]

    % Since the total primal cost is bounded by $\ALG \le 2\sum_{t \in [T]} f_t(x_t)$, it follows that:
    % \[
    %     \ALG ~\le~ 2 \sqrt{\eta}\cdot (\Dual + m) ~\le~ 2 \sqrt{\eta} \cdot \OPT + 2 \sqrt{\eta} \cdot (m+k)
    %     ~.
    % \]

    Combining the above with $\Dual \le \Dual^\star$, $\Dual^\star \le \Primal^\star + k$ (\Cref{lem:primal-dual-objective}) and $\Primal^\star \le \OPT$ (\Cref{lem:primal-objective}), we get that:
    \[
        \Primal ~\le~ 2 \sqrt{\eta}\cdot (\OPT + m + k)
        ~.
    \]

    With $\eta = \max\{1, \ln(k+1)\}$, the competitive ratio is $O(\sqrt{\log k})$.
    % The additive term $ 2 \sqrt{\eta}\cdot (m+k)$ depends only on the system dimensions and is independent of the time horizon $T$, properly satisfying the definition of competitive ratio.
\end{proof}

\subsection{Online Rounding}
\label{sec:online-rounding}

Regularized Greedy only provides a fractional configuration $x_t \in \R_{\ge 0}^m$ at each step.
To obtain a valid deployment strategy, we must convert this into a feasible integral configuration $\bar{x}_t \in \Z_{\ge 0}^m$ satisfying $\sum_{i \in [m]} \bar{x}_{t,i} = k$, while ensuring that the total cost does not increase significantly.

In fact, the online rounding algorithm for the paging and caching problem \citep{BlumBK-FOCS-1999,BansalBN-SICOMP-2012} can be adapted to our setting.
Their algorithms, with appropriate changes, can simultaneously guarantee that (1) $\bar{x}_{t,i} \in \{\lfloor x_{t,i} \rfloor, \lceil x_{t,i} \rceil\}$ and $\E [ \bar{x}_{t,i} ] = x_{t,i}$ for all $i$, and (2) the expected reconfiguration cost is at worst double, i.e., $\mathbb E[\|\bar{x}_t - \bar{x}_{t-1}\|_1] \le 2 \|x_t - x_{t-1}\|_1$.
These conditions ensure that the objective of the rounded solution is at most twice that of the fractional solution.
These online rounding algorithms, however, are not computationally efficient, even in their original paging and caching settings.

Instead, we utilize the smoothness of $f_t$ to design a simpler polynomial-time online rounding procedure that loses at worst a constant factor compared to the fractional solution.

\begin{algorithm}{Lazy Threshold Rounding}
    At the beginning, sample thresholds $\theta_i \in [0,1)$ independently and uniformly for $i \in [m]$.\\[2ex]
    For each time step $t = 0, 1, \ldots, T$:
    \begin{itemize}
        \item Let $\ell_{t,i} = \max\bigl\{ 0, \lfloor x_{t,i}-\theta_i \rfloor \bigr\}$ for $i \in [m]$.
        \item If $t = 0$, let $\bar{x}_0 = x_0$.
        \item Otherwise, let $\bar{x}_t$ be a feasible solution of:
              \[
                  \ell_{t,i} \leq \bar{x}_{t,i} \leq \max \{ \ell_{t,i}, \bar{x}_{t-1,i} \} \quad \forall i \in [m], \qquad \sum_{i \in [m]} \bar{x}_{t,i} = k
                  ~.
              \]
    \end{itemize}
\end{algorithm}

The definition of $\bar{x}_t$ for $t \ge 1$ is feasible because the lower bounds sum to at most $\sum_i x_{t,i} = k$, and the upper bounds sum to at least $\sum_i \bar{x}_{t-1,i} = k$.

% \begin{lemma}
%     \label{lem:rounding-feasibility}
%     At every step $t$, the online rounding step has a feasible integral solution $\bar{x}_t$, and such a solution can be found in polynomial time.
% \end{lemma}

% \begin{proof}
%     The claim is immediate for $t=0$ because $x_0$ is integral and has total mass $k$.
%     For $t \ge 1$, let $u_{t,i} \defeq \max\{\ell_{t,i}, \bar{x}_{t-1,i}\}$.
%     Since $\ell_{t,i} \le x_{t,i}$ for every $i$, we have $\sum_i \ell_{t,i} \le k$.
%     Since $u_{t,i} \ge \bar{x}_{t-1,i}$ for every $i$, we have $\sum_i u_{t,i} \ge k$.
%     Thus there is an integral vector between the lower bounds $\ell_t$ and upper bounds $u_t$ with total mass $k$.
%     It can be found greedily by starting from $\ell_t$ and adding the remaining mass to coordinates that have not reached their upper bounds.
% \end{proof}

\begin{lemma}
    \label{lem:rounding-competitive}
    Lazy Threshold Rounding runs in polynomial-time, and ensures that at every step $t$:
    \begin{enumerate}
        \item $f_t(\bar{x}_t)\le 3 f_t(x_t)$; and
        \item $\mathbb E\left[\|\bar{x}_t - \bar{x}_{t-1}\|_1\right] \le \|x_t - x_{t-1}\|_1$.
    \end{enumerate}
\end{lemma}

\begin{proof}
    \emph{(Running Time)~}
    We can greedily start from the upper bounds, and repeatedly reduce a variable strictly greater than its lower bound by $1$ until the total mass is $k$.

    \bigskip

    \noindent
    \emph{(Latency Cost)~}
    For every expert $i$, the rounding rule gives:
    \[
        \bar{x}_{t,i} \ge \ell_{t,i}
        =
        \max\{0,\lfloor x_{t,i}-\theta_i\rfloor\}
        \ge
        \max\{0,x_{t,i}-2\}
        ~.
    \]
    Hence, $3(1+\bar{x}_{t,i}) \ge 1+x_{t,i}$.
    By \Cref{lem:property-of-f}, we have $f_t(\bar{x}_t) \le 3 f_t(x_t)$.

    \bigskip

    \noindent
    \emph{(Reconfiguration Cost)~}
    The algorithm maintains $\bar{x}_{s,i}\ge \ell_{s,i}$ for every step $s$ and expert $i$.
    Fix a coordinate $i$.
    Since $\bar{x}_{t-1,i} \ge \ell_{t-1,i}$, the rounding rule implies
    \[
        (\bar{x}_{t,i}-\bar{x}_{t-1,i})_+
        \le
        (\ell_{t,i}-\ell_{t-1,i})_+
        ~.
    \]

    If $x_{t,i}\le x_{t-1,i}$, then $(\ell_{t,i}-\ell_{t-1,i})_+=0$.
    Otherwise, $(\ell_{t,i}-\ell_{t-1,i})_+$ is at most the number of thresholds $\theta_i + j$ for $j \in \Z_{\ge 0}$ that are crossed by the interval $(x_{t-1,i},x_{t,i}]$.
    Since $\theta_i$ is uniform over one full period $[0,1)$, the expected
    number of crossed thresholds is at most the interval length:
    \[
        \mathbb E\left[(\bar{x}_{t,i}-\bar{x}_{t-1,i})_+\right]
        \le
        (x_{t,i}-x_{t-1,i})_+
        ~.
    \]

    Summing over $i$ proves the desired bound for incremental cost, which is exactly half the total reconfiguration cost given that the total allocation is fixed at $k$ at all times.
\end{proof}

% \begin{theorem}
%     The Regularized Greedy Algorithm, combined with the online rounding, is $O(\sqrt{\log k})$-competitive against the optimal integral offline adversary.
% \end{theorem}

% \begin{proof}
%     Let $\ALG$ be the expected total cost of the rounded integral algorithm.
%     By \Cref{lem:service-round,lem:movement-round}, we have
%     \begin{align*}
%         \ALG & = \sum_{t=1}^T \Big( \mathbb{E}[f_t(\bar{x}_t)] + \mathbb{E}[\| \bar{x}_t - \bar{x}_{t-1} \|_1] \Big) \\
%              & \le \sum_{t=1}^T \Big( 3 f_t(x_t) + \| x_t - x_{t-1} \|_1 \Big)                                       \\
%              & = \sum_{t=1}^T \Big( 3 f_t(x_t) + 2\sum_{i \in [m]} y_{t,i} \Big)                                     \\
%              & \le 5\sum_{t=1}^T f_t(x_t) ~,
%     \end{align*}
%     where the last inequality uses \Cref{lem:balance-cost}.
%     The fractional analysis gives
%     \[
%         \sum_{t=1}^T f_t(x_t)
%         \le
%         \sqrt{\eta}\cdot(\OPT+k+m)
%         ~.
%     \]
%     Therefore
%     \[
%         \ALG
%         \le
%         5\sqrt{\eta}\cdot(\OPT+k+m)
%         ~.
%     \]
%     With $\eta = \max\{1,\ln(k+1)\}$, the competitive ratio is $O(\sqrt{\log k})$.
% \end{proof}

\section{Hardness for Online Dual Problem}
\label{sec:hardness}

This section proves that the competitive ratio in \Cref{thm:dynamic} is asymptotically tight for the dual problem.
Recall the latency cost $f_t$ and its signed Fenchel conjugate $\hat{f}_t$:
\[
    f_t(x_t) = \max_i \frac{r_{t,i}}{x_{t,i} + 1}
    ~,
    \qquad
    \hat{f}_t(\gamma_t) ~=~ \inf_{x_t \ge 0} \big( f_t(x_t) + \langle x_t , \gamma_t \rangle \big)
    ~.
\]

\paragraph{Online Dual Problem.}
At each time step $t\in[T]$, after observing the workload vector $r_t \in \R^{m}_{\geq 0}$, the algorithm chooses:
\[
    \gamma_t\in\R_{\ge0}^m
    ~,
    \qquad
    \alpha_t\ge0
    ~,
    \qquad
    \beta_{t+1}\in[0,1]^m
    ~,
\]
where $\beta_t$ is already fixed from the previous step $t-1$. The chosen variables
must satisfy:
\[
    \gamma_{t,i}+\beta_{t+1,i}-\beta_{t,i}\le \alpha_t
    \qquad
    \forall i\in[m]
    ~.
\]

The goal of this problem is to maximize the accumulated dual value:
\[
    \Dual
    =
    \sum_{t\in [T]} \widehat f_t(\gamma_t)
    -
    \sum_{t\in [T]} k\alpha_t
    ~.
\]

% offline dual optimum $\Dual^\star$, is the maximum value of the same dual program when the whole request sequence is known in advance.

% The online dual construction problem is the certificate-generation step of an online primal-dual analysis. 
% A primal-dual algorithm proves competitiveness by charging the online primal cost to the value of a dual solution constructed online; hence the best ratio available to this proof framework is constrained by the online solvability of the dual program itself. 
The online primal-dual analysis in \Cref{thm:dynamic} maintains feasible solutions for both primal and dual, with objective values differing by at most an $O(\sqrt{\log k})$ factor.
Hence, the $O(\sqrt{\log k})$ competitive ratio applies to the primal and dual problems alike.
% Since we use online primal dual framework solving both primal and dual program online in \Cref{thm:dynamic}, there is a $O(\ln k)$ competitive algorithm solving the online dual problem.
The theorem below shows that $\sqrt{\log k}$ is an inherent information-theoretic barrier for the online dual problem, even if we allow computationally unlimited algorithms.
% This should be read as a limitation of this online dual-fitting framework, not as an algorithm-independent lower bound for the original primal problem.

\begin{theorem}
    \label{thm:online-barrier}
    There is no online algorithm for the online dual problem with competitive ratio better than $\frac{1}{32}\sqrt{\ln k}$.
\end{theorem}
% The lower bound can be interpreted as a stopping-time budget-allocation problem. In a phase requesting a fixed expert, the online construction decides how much dual mass $\gamma_t$ to place at each time. The benefit of this mass is concave, roughly $\sqrt{\gamma_t}$. The variable $\beta_{t,i}$ provides one unit of free total mass over the phase, since its range is bounded by $[0,1]$. The slack variable $\alpha_t$ can create additional mass, but each unit costs $k$ in the objective.

It suffices to consider deterministic algorithms.
The dual problem has a convex feasible region and a concave maximization objective, so randomization cannot improve the expected value beyond the best deterministic action for a fixed history.
Against deterministic algorithms, without loss of generality we can consider an adaptive adversary, who at each time observes the dual variables already chosen by the online algorithm and decides the workload vector.

We prove the hardness with $m = k+1$ experts, assuming $k \ge 1024$, since excluding finitely many small $k$ does not change the asymptotic bound.

\paragraph{Single Phase Instance.}
We start by describing a gadget that produces a $\sqrt{\log k}$ multiplicative gap between the online and offline dual objectives.
Repeating this gadget for sufficiently many times yields the lower bound even with the additive term in the definition of competitive ratio.

Fix an expert $q$.
The instance repeatedly requests only $q$ for $\tau$ steps where $\tau$ will be determined adaptively according to the algorithm's decisions.
% We state the single phase instance as consecutive requests only to expert $q$. 
Formally, for every time step $t$ in this instance:
\[
    r_{t,i} = \begin{cases}
    1 & \mbox{if $i=q$};              \\
    0 & \mbox{otherwise.}
    \end{cases}
\]

Since only expert $q$ has non-zero workload, the latency cost function and signed Fenchel conjugate function can be written as:
\[
    f_t(x_t)=\frac{1}{1+x_{t,q}} ~,\qquad
    \hat{f}_t(\gamma_t) =
    \begin{cases}
        2\sqrt{\gamma_{t,q}} - \gamma_{t,q}, & \gamma_{t,q}\in[0,1] ~, \\[0.5ex]
        1,                                   & \gamma_{t,q}\ge 1 ~.
    \end{cases}
\]
In either case, we have $\hat{f}_t(\gamma_t)\le 2\sqrt{\gamma_{t,q}}$, which we will use in the arguments below.

Define the objective value up to time $t$ as:
\[
    D_t \defeq \sum_{t' \in [t]} \widehat f_{t'}(\gamma_{t'})
    -
    \sum_{t' \in [t]} k\alpha_{t'}
\]

Let the stopping time $\tau$ be the first step $\tau \ge 4$ satisfying:
% Denote the online value as $D_\tau$ at time step $\tau$. Once $\tau \ge 4$ and 
\[
    D_\tau
    \le
    16 \sqrt{\frac{\tau}{\ln k}}
    ~.
\]
% , the phase stops at time $\tau$.

% Since only expert $q$ is requested, the dual reward depends only on the coordinate $\gamma_{t,q}$. Therefore, in the single-phase analysis, we write
% \[
%     \gamma_t \equiv \gamma_{t,q}
%     ~.
% \]
% Under this notation,
% \[
%     \widehat f_t(\gamma_t)
%     =
%     2\sqrt{\gamma_t}-\gamma_t
%     \qquad
%     \text{for }0\le \gamma_t\le1
%     ~.
% \]

% Consider the case when $n=k + 1$. 

\begin{lemma}
    \label{lem:dual-algo}
    For any online algorithm,
    % the stopping time $\tau\in [4,K]$ is existed for a single phase instance.
    the stopping time $\tau$ is well-defined.
\end{lemma}

\begin{proof}
    % We now define the adversary's stopping rule. After each time $T\ge 4$, 
    We prove a stronger claim that $\tau$ is at most:
    \[
        K=\left\lfloor \frac{k^2}{64}\right\rfloor
        ~.
    \]

    % We claim this rule must stop by time $K$. Suppose not. 
    % Suppose that the requested expert is $q$ in this single phase.
    Assume for contradiction that for every $\tau\in[4,K]$,
    \[
        D_\tau > 16 \sqrt{\frac{\tau}{\ln k}}
        ~.
    \]

    We next examine $\sum_{t \in [K]} \gamma_{t,q}$ and derive contradictory lower and upper bounds for it.

    \paragraph{Lower Bound.}
    The online algorithm needs to stay competitive for every $\tau \in [4, K]$, and hence is forced to overcommit early when $\tau$ is small.
    We next argue that this overcommitment implies a lower bound on $\sum_{t\in [K]} \gamma_{t,q}$.
    % We first show that the total mass $\sum_{t\in [K]} \gamma_{t,q}$ must be large if $D_K$ satisfy the desired bound above.
    % 
    Since $\widehat f_t(\gamma_t)\le 2\sqrt{\gamma_{t,q}}$ and $\alpha_t \ge 0$, we have $D_\tau \le 2\sum_{t \in [\tau]} \sqrt{\gamma_{t,q}}$ for $\tau \in [4, K]$.
    Therefore:
    \begin{equation}
        \label{eqn:sqrt-gamma-bound}
        \sum_{t \in [\tau]} \sqrt{\gamma_{t,q}}
        >
        8\sqrt \frac{\tau}{\ln k}
        ~.
    \end{equation}

    By Hardy's inequality:
    \[
        \sum_{t \in [K]} \gamma_{t,q}
        \ge
        \frac14
        \sum_{t \in [K]}
        \left(\frac{\sum_{t'\in[t]} \sqrt{\gamma_{t',q}}}{t}\right)^2
        ~.
    \]

    Combining with \Cref{eqn:sqrt-gamma-bound}, we obtain:
    \[
        \begin{aligned}
            \sum_{t \in [K]}\gamma_{t,q}
             & \ge
            \frac{16}{\ln k}
            \sum_{t=4}^K\frac1t
            ~.
        \end{aligned}
    \]

    Using $k\ge1024$ and $K=\lfloor k^2/64\rfloor$:
    \[
        \sum_{t=4}^K\frac1t
        \ge
        \int_{4}^{K+1}\frac{dx}{x}
        =
        \ln(K+1)-\ln 4
        \ge
        \ln k
        ~.
    \]

    Thus, we conclude that:
    \[
        \sum_{t \in [K]}\gamma_{t,q}\ge16
        ~.
    \]

    \paragraph{Upper Bound.}
    We next show that $\sum_{t \in [K]} \gamma_{t,q}$ cannot be too large, otherwise $D_K$ will be too small because of the $-k \sum_{t \in [K]} \alpha_t$ term.
    % We next show the upper bound of $\sum_{t=1}^{K} \gamma_t$.
    Since $\widehat f_t(\gamma_t) \le 2\sqrt{\gamma_{t,q}}$, $D_K \le 2\sum_{t \in [K]} \sqrt{\gamma_{t,q}}- k\sum_{t \in [K]} \alpha_t$. By Cauchy--Schwarz, we have:
    \[
        \begin{aligned}
            D_K
            \le
            2\sqrt{K\sum_{t \in [K]} \gamma_{t, q}}-
            k \sum_{t \in [K]} \alpha_t
            ~.
        \end{aligned}
    \]

    Summing the online constraints on the requested coordinate gives
    \[
        \sum_{t \in [K]}\gamma_{t, q}
        \le
        1+ \sum_{t \in [K]} \alpha_t
        ~,
    \]
    because the total decrease of $\beta$ is at most $1$. This implies:
    \[
        D_K
        \le
        2\sqrt{K\sum_{t \in [K]}\gamma_{t, q}}-
        k \left(\sum_{t \in [K]}\gamma_{t, q} - 1\right)
        ~.
    \]

    By $D_K > 0$ and the fact that $\sqrt{K} \le \frac{k}{8}$, we have
    $4\sum_{t \in [K]}\gamma_{t, q}-\sqrt{\sum_{t \in [K]}\gamma_{t, q}}-4<0$.
    Therefore:
    \[
        \sum_{t \in [K]}\gamma_{t,q}
        <
        \frac43
        ~.
    \]

    Comparing the lower and upper bounds shown above, we have a contradiction.
    %
    % Therefore the adaptive stopping rule must stop at some $T\in[4,K]$. At that
    % stopping time,
    % \[
    %     D_T
    %     <
    %     16 \sqrt{\frac{T}{\ln k}}
    %     ~.
    % \]
\end{proof}

\begin{lemma}
    \label{lem:dual-benchmark}
    Consider a single phase instance. Suppose $\beta_{1,q}=1$ in the beginning of the instance. There is an offline strategy simultaneously satisfying:
    \begin{enumerate}
        \item The dual objective value is at least $\sqrt{\tau}$; and
        \item Every expert $j \ne q$ has its $\beta$-value increased by $1/k$ (capped at $1$) in the end.
    \end{enumerate}

\end{lemma}
\begin{proof}
    Assign the unit $\beta$-budget of the requested expert $q$ evenly across the $\tau$ time steps:
    \[
        \gamma_{t,q}=\frac1\tau
        ~,\quad
        \beta_{t+1,q} = 1 - \frac{t}{\tau}
        ~,
        \qquad t \in [\tau]
        ~.
    \]

    For every inactive expert $j\ne q$, the strategy places no dual mass and uses the first step only to increase the $\beta$-value by $1/k$:
    \[
        \gamma_{t,j} = 0
        ~,
        \beta_{t+1,j}
        =
        \min\left\{
        1,\,
        \beta_{1,j}+\frac1k
        \right\},
        \qquad t=1,\ldots,\tau
        ~.
    \]

    The increase is made feasible by setting:
    \[
        \alpha_1=\frac1k
        ~,
        \qquad
        \alpha_t=0 \quad t\in [2,\tau]
        ~,
    \]
    and the corresponding objective penalty is $k\alpha_1=1$.
    Since $\gamma_{t,i} \le \beta_{t,i}-\beta_{t+1,i} + \alpha_t$ is satisfied for every time step $t \in [\tau]$ and every expert $i \in [m]$, this assignment of the dual variables is feasible.
    The offline phase value is at least
    \[
        \begin{aligned}
            \tau\widehat f\big(\frac1\tau\big)-k\alpha_1
             & =
            2\sqrt \tau-2 \ge \sqrt{\tau}
            ~.
        \end{aligned}
    \]
\end{proof}

\begin{proof}[Proof of \Cref{thm:online-barrier}]
    We need to prove that for every constant $C > 0$, there exists an adaptive request sequence $\sigma$ such that
    \[
        \Dual^\star(\sigma)
        \ge
        \frac{\sqrt{\ln k}}{32}
        \Dual(\sigma)
        +
        C
        ~.
    \]

    % \paragraph{Multiple Phases Instance.}
    The adversary repeatedly requests for $R = \lceil C\rceil$ phases where each phase is constructed adaptively following the single phase instance. The requested expert in each phase follows the fixed cyclic order:
    \[
        1,2,\ldots,k+1,1,2,\ldots .
    \] 

    For each single phase $r$, let $D_r$ be the online dual value accumulated during
    that phase, and let $D^{\star}_r$ be the offline optimal value.
    At the beginning, $\beta_1=\textbf{1}$. For each phase $r$, we lower-bound the offline optimum using the strategy of \Cref{lem:dual-benchmark}. The offline strategy uses the first step of each phase to increase every non-requested expert's $\beta$-value by $1/k$, capped at $1$. Therefore, after an expert $q$ is requested, it is not requested again for the next $k$ phases. During these $k$ phases, $q$ receives $k$ increments of its $\beta$-value with size $1/k$. Hence, by the next time $q$ is requested, its $\beta$-value has returned to $1$. Thus the strategy is feasible for the multiple-phase instance. Combining with \Cref{lem:dual-algo}, we obtain $D^{\star}_r \ge \frac{\sqrt{\ln k}}{16}D_r$ and $D^{\star}_r \ge 2$. Averaging these two bounds gives
    \[
        D^{\star}_r
        \ge
        \frac{\sqrt{\ln k}}{32}D_r+1
        ~.
    \]

    Summing up $R$ phases, we obtain
    \[
        \begin{aligned}
            \Dual^\star(\sigma)
            \ge
            \frac{\sqrt{\ln k}}{32}
            \Dual(\sigma)
            +
            C
            ~.
        \end{aligned}
    \]

    This completes the proof.
\end{proof}

\section{Static Model}

In the Static MoE Serving problem, the configuration remains fixed across all time steps, meaning $x_t = x \in \Z_{\ge 0}^m$ for all $t \in [T]$.
The primal convex program simplifies to minimizing the total latency cost $\sum_{t \in [T]} f_t(x)$ subject to $\sum_{i \in [m]} x_i = k$ and $x_i \ge 0$.
That is:
\begin{align*}
    \textrm{minimize} \quad
              &
    \sum_{t \in [T]} f_t(x)           \\
    \textrm{subject to} \quad
              &
    \sum_{i \in [m]} x_i = k          \\
              &
    x_i \ge 0 &   & \forall i \in [m]
\end{align*}

Let $x^\star \in \R_{\ge 0}^m$ denote the optimal fractional solution to this static relaxation.

\subsection{Rounding Algorithm and Approximation Ratio}

We propose a simple deterministic rounding algorithm to convert the fractional optimal solution $x^\star$ into a feasible integral configuration.

\begin{algorithm}{Rounding Algorithm}
    \begin{itemize}
        \item Solve the static convex program relaxation to obtain the optimal fractional configuration $x^\star \in \R_{\ge 0}^m$.
        \item For each expert $i \in [m]$, round down the allocation to construct an integral configuration $\bar{x}_i = \lfloor x^\star_i \rfloor$.
        \item Allocate the remaining $k - \sum_{i \in [m]} \bar{x}_i$ arbitrarily.
    \end{itemize}
\end{algorithm}

\begin{theorem}
    The Rounding Algorithm is a $2$-approximation for the Static MoE Serving problem.
\end{theorem}

\begin{proof}
    Let $\bar{x} \in \Z_{\ge 0}^m$ be the configuration produced by the algorithm. By construction, $\bar{x}$ is a feasible configuration.

    For any time step $t \in [T]$, the latency cost under configuration $\bar{x}$ is:
    \[
        f_t(\bar{x}) ~=~ \max_{i \in [m]} \frac{r_{t,i}}{1 + \bar{x}_i} ~\le~ \max_{i \in [m]} \frac{r_{t,i}}{1 + \lfloor x^\star_i \rfloor} ~.
    \]

    Since $1 + x \le 2(1+\lfloor x \rfloor)$, the above is at most:
    \[
        2 \cdot \max_{i \in [m]} \frac{r_{t,i}}{1+x^\star_i} ~=~ 2 \cdot f_t(x^\star)
        ~.
    \]

    Summing over $t \in [T]$, we have $\ALG \le 2 \cdot \Primal \le 2 \cdot \OPT$.
\end{proof}

\subsection{Integrality Gap}

Next, we demonstrate that the above rounding for the convex program relaxation is tight, as the integrality gap approaches exactly $2$ in the worst case.

\begin{theorem}
    The integrality gap of the static convex program relaxation is at least $2 - \frac{1}{m}$, which approaches $2$ as $m \to \infty$.
\end{theorem}

\begin{proof}
    Consider an instance with $m$ experts and $n = 2m - 1$ total GPUs, which means the budget is $k = m - 1$.
    Let there be a single time step ($T = 1$) with a uniform workload vector $r_{1,i} = 1$ for all $i \in [m]$.

    For the fractional relaxation, a symmetric allocation evenly distributes the budget: $x^\star_i = \frac{k}{m} = 1 - \frac{1}{m}$ for all $i \in [m]$.
    The latency cost of this fractional solution is:
    \[
        \Primal ~\le~ f_1(x^\star) ~=~ \frac{1}{1 + \left(1 - \frac{1}{m}\right)} ~=~ \frac{m}{2m - 1} ~.
    \]

    Now, consider any feasible integral configuration $x \in \Z_{\ge 0}^m$.
    By the Pigeonhole Principle, at least one expert $i^* \in [m]$ must receive zero additional GPUs beyond the first one, meaning $x_{i^*} = 0$.
    The latency cost is bottlenecked by the expert with the minimum allocation:
    \[
        \OPT ~\ge~ \frac{1}{1 + x_{i^*}} ~=~ 1 ~.
    \]

    The integrality gap is then at least $\frac{2m-1}{m} = 2 - \frac{1}{m}$ as stated in the theorem.
\end{proof}

\subsection{Hardness of Approximation}

Having established a $2$-approximation and a matching integrality gap for the static convex program relaxation, a natural theoretical question is whether we can bypass this relaxation barrier to achieve an arbitrarily good approximation.
We answer this negatively.
More precisely, we rule out the possibility of a Fully Polynomial-Time Approximation Scheme (FPTAS) under the Exponential Time Hypothesis (ETH).

Our reduction relies on the \emph{Densest $k$-Subgraph} (\DkS) problem:

\begin{quote}
    \em
    Given an undirected graph $G=(V, E)$ and an integer $k \le |V|$, find a subset of vertices $S \subseteq V$ of size $k$ that maximizes the number of induced edges $|E(S)|$.
\end{quote}

The hardness of \DkS has been extensively studied.
For a vertex subset $S$ with $|S| = k$, the \emph{density} of the induced subgraph is $|E(S)| / \binom{k}{2} \in [0,1]$, so a $k$-clique has density $1$.
We leverage the following inapproximability result.

\begin{lemma}[{\citealp[Theorem~1]{Manurangsi-STOC-2017}}]
    \label{lem:dks-hardness}
    Assuming ETH, there is a constant $c>0$ such that no polynomial-time algorithm can distinguish between the following two cases for a graph $G$ with $m$ vertices and a positive integer $k \le m$:
    \begin{enumerate}
        \item \textbf{Clique:~} There exists a $k$-clique.
        \item \textbf{Sparse:~} Every $k$-subgraph of $G$ has density at most $\delta = m^{-1/(\log \log m)^c}$.
    \end{enumerate}
\end{lemma}

Using this result, we next establish the hardness of the Static MoE Serving problem.

\begin{theorem}
    The Static MoE Serving problem is NP-hard. Furthermore, assuming the Exponential Time Hypothesis (ETH), it does not admit an FPTAS.
\end{theorem}

\begin{proof}
    We prove this by a reduction from \DkS.
    Given an instance $(G, k)$ of \DkS where $G=(V,E)$ and $m = |V|$, we construct an instance of Static MoE Serving as follows:
    \begin{itemize}
        \item \textbf{Experts:~} Let there be an expert for each vertex of $G$.
              There are $m = |V|$ experts.
        \item \textbf{Resources:~} Let the total number of GPUs be $n = m + k$.
        \item \textbf{Requests:} Let there be a time step for each edge $e = (u, v)\in E$.
              We will abuse notation and refer to the time step also as $e$.
              Let the workload vector of step $e$ be:
              \[
                  r_{e, i} =
                  \begin{cases}
                  2 & \text{if $i \in \{u, v\}$;} \\
                  1 & \text{otherwise.}
                  \end{cases}
              \]
    \end{itemize}

    Let $x \in \Z_{\ge 0}^m$ be a feasible configuration, i.e., $\sum_{i \in [m]} x_i = k$.
    We may assume without loss of generality that $x_i \in \{0, 1\}$ for all $i$.
    This is because assigning $x_i \ge 2$ does not further decrease the objective, since the maximum workload is $2$ --- the latency cost would be bottlenecked by an expert without additional GPUs.

    Let $S = \{i \mid x_i = 1\}$ be the set of experts receiving an extra GPU. By the budget constraint, $|S|=k$.
    The latency cost for a step $e = (u, v)$ is:
    \[
        f_e(x) = \max_{i \in [m]} \frac{r_{e,i}}{1+x_i} ~.
    \]

    If both endpoints are selected, i.e., $x_u = x_v = 1$, we have $\frac{r_{e,u}}{1+x_u} = \frac{r_{e,v}}{1+x_v} = \frac{2}{2} = 1$. For any other expert $w \notin \{u,v\}$, the cost is at most $\frac{1}{1+0} = 1$. Thus, $f_e(x) = 1$.
    By contrast, if at least one endpoint is not selected, the maximum evaluates to $\frac{2}{1+0} = 2$, yielding $f_e(x) = 2$.

    Therefore, the total objective value is determined by the number of induced edges $|E(S)|$:
    \[
        \ALG(S) ~=~ 1 \cdot |E(S)| + 2 \cdot (|E| - |E(S)|) ~=~ 2|E| - |E(S)| ~.
    \]

    \paragraph{NP-hardness.}
    Minimizing the Static MoE cost is exactly equivalent to maximizing $|E(S)|$, the \DkS objective. Since \DkS is NP-hard, Static MoE Serving is also NP-hard.

    \paragraph{Impossibility of FPTAS.}
    Mapping the hardness of \DkS from \Cref{lem:dks-hardness} to our problem, the optimal cost $\OPT$ behaves as follows:
    \begin{itemize}
        \item In the \textbf{Clique} case, the optimal cost is $\OPT \le 2|E| - \binom{k}{2}$.
        \item In the \textbf{Sparse} case, the optimal cost is $\OPT \ge 2|E| - \delta \binom{k}{2}$.
    \end{itemize}

    Suppose, for contradiction, that there is an FPTAS for Static MoE Serving. For any $\epsilon > 0$, the FPTAS produces a solution with cost at most $(1+\epsilon)\OPT$ in $\mathrm{poly}(m, 1/\epsilon)$ time.

    To successfully distinguish the two cases in polynomial time, we only need to choose $\epsilon$ small enough such that the approximation intervals of the two cases do not overlap. That is, the worst-case approximation of the Clique case must be strictly less than the best-case cost of the Sparse case:
    \[
        (1+\epsilon) \left( 2|E| - \binom{k}{2} \right) ~<~ 2|E| - \delta \binom{k}{2} ~.
    \]

    This inequality simplifies to:
    \[
        2\epsilon |E| ~<~ (1+\epsilon-\delta)\binom{k}{2} ~.
    \]

    Since the graph has $m$ vertices, the total number of edges is bounded by $|E| \le \binom{m}{2}$. Furthermore, $\delta = o(1) \le \frac{1}{2}$ for sufficiently large $m$. Therefore, it suffices to choose an $\epsilon$ satisfying:
    \[
        \epsilon \le \frac{k(k-1)}{4m(m-1)} ~.
    \]

    This means that $1/\epsilon$ is polynomially bounded by $O(m^2/k^2)$. Consequently, the running time of the FPTAS would be polynomial in $m$, giving us a polynomial-time algorithm to distinguish the Clique case from the Sparse case. This contradicts the hardness result of \citet{Manurangsi-STOC-2017}.
\end{proof}

\appendix

\section{Deferred Proof from Section~\ref{sec:convex-programs}}
\label{app:convex-programs}

\subsection{Proof of Lemma~\ref{lem:primal-objective}}

The integer constraints in the original problem require
\(x_{t,i}\in \mathbb Z_{\ge 0}\) and \(\sum_i x_{t,i}=k\). The convex program
relaxes these requirements to \(x_{t,i}\in \mathbb R_{\ge 0}\) and
\(\sum_i x_{t,i}\le k\). Since the objective function $f_t$ is convex and the constraints are linear, the minimum over the continuous relaxed domain is less than or equal to the minimum over the discrete integer domain.

\subsection{Proof of Lemma~\ref{lem:property-of-f}}

\emph{(Conjugate pair)}
Let $A_t(x_t)$ be the set of active coordinates attaining the maximum in $f_t(x_t)$.
By the subgradient convention above, every $p_t\in\partial f_t(x_t)$ can be written as
\[
    p_{t,i}
    =
    -\sum_{j\in A_t(x_t)}
    \lambda_j
    \frac{r_{t,j}}{(1+x_{t,j})^2}
    \mathbf 1_{\{i=j\}}~,
\]
where $\lambda_j\ge 0$ and $\sum_{j\in A_t(x_t)}\lambda_j=1$.
If $\gamma_t=-p_t$, then:
\[
    \left\langle \gamma_t,x_t+\mathbf 1\right\rangle
    = \sum_{j\in A_t(x_t)} \lambda_j \frac{r_{t,j}}{1+x_{t,j}}
    = \sum_{j\in A_t(x_t)} \lambda_j f_t(x_t)
    = f_t(x_t) ~.
\]

\bigskip

\noindent
\emph{(Smoothness of $f_t$)}
Suppose that $\theta\ge 1$ and $x_{t,i}+1\le \theta(x'_{t,i}+1)$ for every $i\in[m]$.
Then
\[
    \frac{r_{t,i}}{1+x'_{t,i}}
    \le
    \theta \frac{r_{t,i}}{1+x_{t,i}}
    \qquad
    \forall i\in[m].
\]
Taking the maximum over $i$ gives $f_t(x'_t)\le \theta f_t(x_t)$.

\bigskip

\noindent
\emph{(Sublinearity of $\hat{f}_t$)}
By the form of $f_t$ and the signed Fenchel transform:
\[
    \hat f_t(\gamma_t)
    =
    \inf_{x_t \ge 0}
    \left(
    \max_{i\in[m]}\frac{r_{t,i}}{1+x_{t,i}}
    +
    \sum_i \gamma_{t,i} x_{t,i}
    \right).
\]

Introduce an epigraph variable $\tau$ for the maximum term.
Then, the infimum for a fixed $\tau$ is achieved when we let $x_{t,i}$ be as small as possible subject to $\tau \ge r_{t,i}/(1+x_{t,i})$ for every $i$.
Thus:
\[
    \hat f_t(\gamma_t) = \inf_{\tau>0} \left[ \tau + \sum_i \gamma_{t,i} \left( \frac{r_{t,i}}{\tau}-1 \right)_+ \right] ~.
\]

Fix $\theta\in(0,1]$.
For any $\tau=\sqrt{\theta}s$ with $s>0$, we have:
\[
    \begin{aligned}
        \tau + \sum_i \theta \gamma_{t,i} \left( \frac{r_{t,i}}{\tau}-1 \right)_+
         & = \sqrt{\theta}s + \theta \sum_i \gamma_{t,i} \left( \frac{r_{t,i}}{\sqrt{\theta}s}-1 \right)_+              \\
         & = \sqrt{\theta} \left[ s + \sum_i \gamma_{t,i} \left( \frac{r_{t,i}}{s} - \sqrt{\theta} \right)_+ \right] ~.
    \end{aligned}
\]

Since $\theta \in (0, 1]$, the term between the brackets is at least:
\[
    s + \sum_i \gamma_{t,i} \left( \frac{r_{t,i}}{s} - 1 \right)_+ \ge \hat{f}_t(\gamma_t) ~.
\]

We conclude that:
\[
    \tau + \sum_i \theta \gamma_{t,i} \left( \frac{r_{t,i}}{\tau}-1 \right)_+ \ge \sqrt{\theta}\,\hat f_t(\gamma_t) ~.
\]

Taking infimum of the left-hand side over $\tau>0$ proves $\hat f_t(\theta \gamma_t) \ge \sqrt{\theta}\,\hat f_t(\gamma_t)$.

\subsection{Derivation of the Offset Fenchel Dual and Proof of Lemma~\ref{lem:primal-dual-objective}}

We derive the dual program in \Cref{sec:convex-programs} and prove the additive weak-duality bound in \Cref{lem:primal-dual-objective}.
For each budget constraint, introduce a multiplier $\alpha_t \in \R_{\ge 0}$.
For each incremental reconfiguration constraint $x_{t,i}-x_{t-1,i}-y_{t,i}\le 0$, introduce a multiplier $\beta_{t,i}\ge 0$.
We also introduce an auxiliary terminal multiplier $\beta_{T+1,i}\ge 0$ for the domain constraint $-x_{T,i}\le 0$; this variable lets us write the dynamic terms uniformly.
Finally, introduce $\gamma_t \in \R_{\ge 0}^m$ as a Fenchel variable and use:
\[
    f_t(x_t)
    =
    \sup_{\gamma_t \in \R_{\ge 0}^m}
    \bigl(
    \hat f_t(\gamma_t) - \langle x_t, \gamma_t \rangle
    \bigr)
    ~.
\]

% For any fixed $\gamma_t\ge 0$, the affine term $\hat f_t(\gamma_t)-\langle x_t,\gamma_t\rangle$ is a lower bound on $f_t(x_t)$.
The corresponding Lagrange lower bound is
\begin{align*}
    L(x,y;\alpha,\beta,\gamma)
    ~=~
     & \sum_{t \in [T]} \Bigl( \hat f_t(\gamma_t) - \langle x_t, \gamma_t \rangle \Bigr)
    + \sum_{t \in [T]}\sum_{i \in [m]} y_{t,i}
    + \sum_{t \in [T]} \alpha_t \Bigl( \sum_{i \in [m]} x_{t,i} - k \Bigr)                      \\
     & + \sum_{t \in [T]}\sum_{i \in [m]} \beta_{t,i} \big( x_{t,i} - x_{t-1,i} - y_{t,i} \big)
    - \sum_{i \in [m]} \beta_{T+1,i}x_{T,i} ~.
\end{align*}

For every feasible $(x,y)$, the equality term vanishes, the inequality terms are non-positive, and $\hat f_t(\gamma_t) - \langle x_t, \gamma_t \rangle \le f_t(x_t)$ by definition.
Hence, $\inf_{x,y\ge 0} L(x,y;\alpha,\beta,\gamma)$ is a lower bound on $\Primal^\star$ for all $\alpha \ge 0$, $\beta \ge 0$, and $\gamma \ge 0$.

Grouping terms by the primal variables gives
\begin{align*}
    L(x,y;\alpha,\beta,\gamma)
    ~=~
    \sum_{t \in [T]} \hat f_t(\gamma_t)
     & + \sum_{t \in [T]}\sum_{i \in [m]} x_{t,i}
    \bigl( \alpha_t + \beta_{t,i} - \beta_{t+1,i} - \gamma_{t,i} \bigr)          \\
     & + \sum_{t \in [T]} \sum_{i \in [m]} y_{t,i} \bigl( 1 - \beta_{t,i} \bigr) \\
     & - \sum_{t \in [T]} k \alpha_t - \langle x_0, \beta_1 \rangle ~.
\end{align*}

The minimization over $x_{t,i}\ge 0$ and $y_{t,i}\ge 0$ yields dual constraints:
\[
    \gamma_{t,i} + \beta_{t+1,i} - \beta_{t,i} \le \alpha_t~,
    \qquad
    \gamma_{t,i}\ge 0~,
    \qquad
    0\le \beta_{t,i}\le 1~.
\]

The resulting Fenchel dual program is:
\begin{align*}
    \textrm{maximize} \quad
     &
    \sum_{t \in [T]} \hat{f}_t(\gamma_t) - \sum_{t \in [T]} k \alpha_t - \langle x_0,\beta_1\rangle        \\
    \textrm{subject to} \quad
     & \gamma_{t,i} + \beta_{t+1,i} - \beta_{t,i} \le \alpha_t &  & \forall t \in [T], \forall i \in [m]   \\[1ex]
     & 0 \le \beta_{t,i} \le 1                                 &  & \forall t \in [T+1], \forall i \in [m] \\[1ex]
     & \gamma_{t,i} \ge 0                                      &  & \forall t \in [T], \forall i \in [m] \\[1ex]
     & \alpha_{t} \ge 0                                      &  & \forall t \in [T]
     ~.
\end{align*}

By weak duality, the optimal value of this dual is at most $\Primal^\star$.
The dual program in \Cref{sec:convex-programs} omits the term $-\langle x_0,\beta_1\rangle$ from the objective.
By $0\le \beta_{1,i}\le 1$ for all $i$ and $\|x_0\|_1=k$, we have $0\le \langle x_0,\beta_1\rangle\le k$.
Hence, its optimal value $\Dual^\star$ is at most the optimal of the above dual plus $k$.
This gives $\Dual^\star \le \Primal^\star + k$.

\section{Deferred Proofs from Section~\ref{sec:dynamic}}

\subsection{Proof of Lemma~\ref{lem:balance-cost}}
\label{app:balance-cost}

First, we bound the incremental reconfiguration cost using the inequality $u - v \le u \ln\dfrac{u}{v}$ for $u, v > 0$:
\[
    x_{t,i} - x_{t-1,i} ~\le~ (x_{t,i}+1) \ln \frac{x_{t,i}+1}{x_{t-1,i}+1} ~.
\]

Summing over all $i$ where the allocation strictly increases (i.e., $x_{t,i} > x_{t-1,i} \ge 0$):
\begin{align*}
    \sum_i y_{t,i}
     &
    ~\le~ \sum_{i \,:\, x_{t,i} > x_{t-1,i}} (x_{t,i}+1) \ln \frac{x_{t,i}+1}{x_{t-1,i}+1}
    ~.
\end{align*}

By \Cref{eqn:kkt} with equality (since $x_{t,i} > 0$) and the relation between $\alpha_t, \gamma_{t,i}$ and $\alpha'_t, \gamma'_{t,i}$:
\[
    \sum_i y_{t,i}
    ~\le~ \eta \sum_{i \,:\, x_{t,i} > x_{t-1,i}} (x_{t,i}+1)(\gamma_{t,i} - \alpha_t)
\]

Dropping $\alpha_t \ge 0$ and expanding the summation to all $i \in [m]$ (since $\gamma_{t,i} \ge 0$), we have:
\[
    \sum_i y_{t,i} ~\le~ \eta \sum_{i \in [m]} (x_{t,i}+1) \gamma_{t,i} ~.
\]

Again by the relation between $\gamma_{t,i}$ and $\gamma^{\prime}_{t,i}$, we have:
\[
    \sum_i y_{t,i}
    ~\le~ \sum_{i \in [m]} (x_{t,i}+1) \gamma^{\prime}_{t,i}
    ~=~ f_t(x_t) ~,
\]
where the last equality follows by \Cref{lem:property-of-f}.

\bibliographystyle{plainnat}
\bibliography{moe}

\end{document}

%% file: notation.tex
\newcommand{\Z}{\mathbb{Z}}

\newcommand{\R}{\mathbb{R}}

\newcommand{\ALG}{\textsc{ALG}}
\newcommand{\OPT}{\textsc{OPT}}
\newcommand{\Primal}{\textsc{Primal}}
\newcommand{\Dual}{\textsc{Dual}}

\newcommand{\KL}{D_\mathrm{KL}}

\DeclareMathOperator*{\argmin}{arg\,min}

\newcommand*{\defeq}{\stackrel{\text{def}}{=}}

\DeclareMathOperator{\E}{\mathbf{E}}

\newcommand{\DkS}{DkS\xspace}